\documentclass[preprint]{article}

\usepackage{neurips_2026}

\usepackage[utf8]{inputenc} 
\usepackage[T1]{fontenc}    
\usepackage{hyperref}       
\usepackage{url}            
\usepackage{booktabs}       
\usepackage{amsfonts}       
\usepackage{nicefrac}       
\usepackage{microtype}      
\usepackage{xcolor}         
\usepackage{enumerate}
\usepackage{amsmath, amssymb, amsthm}
\usepackage{multirow}
\usepackage{booktabs}
\usepackage{graphicx}
\usepackage{comment}
\usepackage{enumitem}
\usepackage[skip=0pt]{caption}
\usepackage{wrapfig}
\usepackage{algorithm}
\usepackage{algpseudocode}

\usepackage{makecell}
\usepackage{xcolor}
\definecolor{gainpos}{RGB}{0,60,255} 
\definecolor{gainneg}{RGB}{180,30,30}  
\makeatletter
\newcommand{\gain}[1]{\@gaintest#1\@nil}
\def\@gaintest#1#2\@nil{%
  \ifx-#1%
    \textcolor{gainneg}{\tiny #1#2}%
  \else
    \textcolor{gainpos}{\tiny #1#2}%
  \fi}
\makeatother

\newcommand{\cg}[2]{\makecell{#1\makebox[1.5em][c]{\gain{#2}}}}

\title{Non-negative Elastic Net Decoding for Information Retrieval}

\author{
Koki Okajima\thanks{Corresponding author: \texttt{koki.okajima@ntt.com}} \ 
 \quad \  Yasutoshi Ida 
 \quad \ Tsukasa Yoshida
 \quad \ Yasuaki Nakamura
  \\
  \vspace{-0.8em}
  \\
  NTT, Inc.
  \vspace{-1.0em}
}

\usepackage{thmtools}
\usepackage{thm-restate}

\DeclareMathOperator*{\argmax}{arg\,max}
\DeclareMathOperator*{\argmin}{arg\,min}

\begin{document}

\maketitle

\begin{abstract}
Dense retrieval has become the dominant paradigm in information retrieval, in which each document is scored against a query by the inner product of their vector embeddings, and the top-$k$ documents by score are retrieved for this query. 
However, since each document's score depends solely on the embedding of the query and itself, the retrieval process is oblivious to the content of the entire corpus. Therefore, 
dense retrieval cannot avoid selecting semantically similar documents from the corpus, which may result in a non-diverse, redundant set of retrieved documents. 
To this end, we approach retrieval as a joint decoding problem, in which documents are selected as a set with regard to the context of the rest of the corpus. 
To achieve this, we propose Non-Negative elastic Net (NNN) decoding, which selects documents whose embeddings jointly reconstruct the query embedding as a sparse non-negative linear combination.

Our main theoretical result establishes a strict separation between dense retrieval and NNN decoding. 
For any corpus, every query correctly handled by dense retrieval is also handled by NNN decoding, while on corpora containing correlated documents, NNN decoding additionally handles queries that dense retrieval cannot. 
Experimental results indicate that applying NNN decoding to frozen embeddings trained for inner-product scoring yields consistent improvements across several benchmarks. 
Moreover, we introduce an end-to-end training procedure which optimizes the embeddings for NNN decoding, producing significant performance gains surpassing in all metrics and benchmarks compared to dense retrieval. 
Our work establishes a new paradigm for leveraging dense embeddings in information retrieval, beyond the standard practice of inner-product scoring. 
\end{abstract}

\section{Introduction}

Information retrieval is a ubiquitous procedure in modern information systems, ranging from search engines to modern AI pipelines such as in retrieval-augmented generation \citep{lewis2020rag, guu2020realm, borgeaud2022retro} and agentic systems that select from large tool catalogues \citep{schick2023toolformer, qin2023toolllm, patil2023gorilla}.  
Typically, an information retrieval system aims to solve a needle-in-a-haystack problem, in which it selects a few documents from a large corpus which it redeems as relevant to a query, and discards the rest. 
Therefore, it is important to develop efficient yet precise retrieval systems to extract information content relevant to the query. 

Dense retrieval is the dominant method for achieving low latency in such tasks, in which a bi-encoder maps both the query and each document to a fixed-dimensional vector and scores relevance by inner product or cosine similarity \citep{Huang2013SemanticModels, karpukhin2020dpr}. 
These neural language embedding models, capturing the semantic relationship between corpus and queries, enable matching beyond lexical overlap \citep{wang2022e5, ni2022sentencet5}, which is a key advantage over term-statistics based methods such as TF-IDF and BM25\citep{robertson1995bm25}. These properties have made dense bi-encoders the backbone of production search and retrieval-augmented pipelines, with research mainly conducted on enhancing its training protocol and architecture \citep{wang2022e5, ni2022sentencet5, xiong2021ance, izacard2021contriever}.

\begin{figure}
    \centering
    \includegraphics[width=1.0\linewidth]{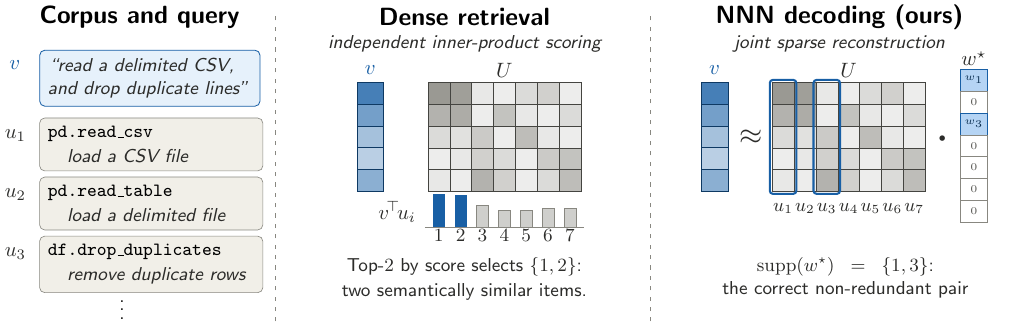}
    \caption{ 
    A conceptual comparison of dense retrieval and NNN decoding on a tool retrieval task. The query requires two tools $\{1, 3\}$, but $u_1$ and $u_2$ are semantically similar to one another, hence having similar embeddings. Dense retrieval (middle) scores each tool independently, and may return the redundant top-2 pair $\{1, 2\}$. NNN decoding (right) jointly reconstructs $v$ as a sparse non-negative combination of the columns of $U$ and recovers the correct, non-redundant pair $\{1, 3\}$.
    }
    \label{fig:concept}
    \vspace*{-12pt}
\end{figure}

Despite the progress in improving bi-encoder models, the inner-product scoring mechanism has remained fixed in dense retrieval. 
Since each document's score depends only on itself and the query, the dense retriever never exploits the fact that documents in the corpus are correlated with one another. 
In other words, dense retrieval uses only vector similarity with the query embedding, higher-ranked documents may have strongly correlated embeddings, typically indicating semantic similarity. 
Therefore, using this ranking procedure may result in a redundant set of documents with high informational overlap, regardless of the overall content of the corpus. 
This is especially problematic for tasks where the objective is to retrieve a complementary set of documents or items to complete a query rather than to accumulate similar information sources, which is the case for tool retrieval and multi-hop retrieval \citep{Carbonell1998MMR, Qu2024CompletenessOriented, Rezazei2025VendiRAG}. 

Overcoming this limitation requires a scoring rule that judges documents collectively, allowing the selection of one document to discount others that overlap with it. Such a collective scoring rule can be expressed by modeling the query embedding as a linear combination of corpus embeddings, which induces a mechanism that reduces the selection of overlapping documents. 
We realize this through our proposed method \emph{\textbf{\underline{N}}on-\textbf{\underline{N}}egative elastic \textbf{\underline{N}}et (NNN) decoding}, a retrieval method which selects documents whose embeddings jointly reconstruct the query embedding under a sparsity-inducing penalty (Figure~\ref{fig:concept}), in which its intensity is governed by hyperparameters. 

The strict gap between dense retrieval and NNN decoding can be realized theoretically. We prove that there exist corpora on which NNN decoding recovers the relevant subset while dense retrieval fails, and such corpora are themselves an instance of the near-duplicate scenario described above (Proposition~\ref{prop:strict}).  More broadly, NNN decoding generalizes inner-product scoring at the query level: on every corpus and every target relevant set, any query correctly handled by inner-product scoring is also handled by NNN decoding with suitable hyperparameters (Theorem~\ref{thm:inclusion}). This establishes NNN decoding as a viable drop-in replacement of inner-product scoring in a theoretically grounded way. 

This difference is also realized empirically on tool retrieval and multi-hop retrieval benchmarks. Applying NNN decoding to the embeddings of an inner-product-trained bi-encoder already improves the probability of retrieving all relevant documents over dense retrieval. This can further be improved by fine-tuning the embeddings tailored to NNN decoding's reconstruction objective, 
resulting in up to a 36\% performance increase in the completeness metric. 
This gap widens as the number of relevant documents per query grows, which matches the regime in which the theory predicts NNN decoding to be most effective; namely, the queries whose relevant set is more likely to contain documents correlated with irrelevant ones that inner-product scoring may have difficulty separating.

\subsection{Contributions} 
\label{sec:contributions}
In summary, our main contributions are as follows: 
\begin{enumerate}[leftmargin=0.5cm] 
    \item We propose NNN decoding, a new information retrieval method that replaces independent inner-product scoring with non-negative elastic net regression (Section~\ref{sec:theory}).
    \item We prove that any query correctly handled by inner-product scoring is correctly handled by elastic net for some choice of hyperparameters (Theorem~\ref{thm:inclusion}), and there exist queries on which elastic net succeeds while inner-product scoring fails (Proposition~\ref{prop:strict}). This proves that NNN decoding enlarges the set of queries on which the correct relevant subset is recovered beyond dense retrieval (Section~\ref{sec:theory}).
    \item We develop an end-to-end training procedure based on Unrolled FISTA to fine-tune the bi-encoder for NNN decoding (Section~\ref{sec:method-unrolled}), and show that it improves retrieval performance on tool retrieval and multi-hop retrieval benchmarks (Section~\ref{sec:experiments}).
    \item We show that NNN decoding applied at test time to the frozen embeddings of an inner-product-trained encoder yields consistent precision gains with only two tunable hyperparameters, demonstrating that the theoretical separation is partially accessible without any retraining (Section~\ref{sec:experiments}), and is a highly effective direct replacement for inner-product scoring. 
\end{enumerate}
\subsection{Related Works}

\paragraph{Dense retrieval.}
Dense bi-encoder retrieval \citep{Huang2013SemanticModels, karpukhin2020dpr} has been improved along several axes, including hard-negative mining \citep{xiong2021ance}, unsupervised contrastive pretraining \citep{izacard2021contriever}, instruction-tuned embeddings \citep{wang2022e5, su2022instructor}, and scaling to larger language model backbones \citep{ni2022sentencet5, muennighoff2024gritlm, lee2025gemini}. 
These methods share a common scoring mechanism, in which each document receives an independent inner-product score against the query and the top-$k$ documents by score are returned. Recently, there have also been work investigating the information-theoretical threshold on the embedding dimension of dense bi-encoders with inner-product scoring to recover any $k$-subset of the corpus \citep{bangachev2025global, weller2026limit, Wang2026R2k, bangachev2026dbarrier, Okajima2026Quant}. 

Our work is orthogonal to this line of research. Rather than competing with advances in encoder architecture or training data, we replace the scoring function with a decoder that can be applied on top of any bi-encoder, including those trained entirely for inner-product relevance. 
Moreover, our theoretical results consider not the capacity of the embedding space under a fixed scoring rule, but the expressivity of the scoring rule itself under fixed embeddings, isolating the gains attributable to the scoring rule rather than to the embedding dimension.

\paragraph{Retrieval architectures beyond single-vector dense scoring.}

A range of architectures improve on single-vector inner-product retrieval while retaining independent scoring. 
Some examples include cross-encoders \citep{nogueira2019passage} which jointly encode the query and each candidate document onto a score number, or multi-vector models such as ColBERT \citep{khattab2020colbert}. In each case the scoring function factorizes across documents, and the top-$k$ is read off a list of independently computed scores.
On the contrary, listwise LLM rerankers \citep{Sun2023RankGPT, Adeyemi2024ListwiseReranker} and generative retrieval \citep{Lee2023GLEN, Li2024GenRetrieval} break this independence by jointly scoring documents fed into the model. However, both incur language-model inference at query time, which may be computationally infeasible for latency-sensitive usages. NNN decoding also breaks this separability but without this computational cost, while employing precomputed embeddings which are capable of indexing a large corpus in a memory-efficient manner. 

\paragraph{Sparse coding and compressed sensing.}
The mathematics of reconstructing a target vector as a combination of atoms drawn from a dictionary, has a long history in signal processing and statistics. The Lasso \citep{Tibshirani1996lasso} and elastic net \citep{Zou2005elasticnet} established regularization terms based on $\ell_1$ penalty as a practical means of inducing sparsity in regression problems. Compressed sensing theory provides recovery guarantees that identify when the support of a sparse signal can be exactly recovered from linear measurements, typically under conditions such as the restricted isometry property or incoherence of the dictionary \citep{candes2006stable, donoho2006cs, Kabashima2009Lpmin, Wainwright2009Lasso, Miolane2021distribution, Okajima2023UltraSparse}. 
Algorithmically, this class of optimization problems can be solved by coordinate descent \citep{Friedman2010Regularization}, ADMM \citep{Boyd2011ADMM}, or proximal-gradient methods \citep{Daubechies2004ISTA}. Our NNN decoding algorithm build on FISTA \citep{beck2009fista}, which parallelizes efficiently on modern GPUs and attains the optimal convergence rate among first-order methods for minimizing smooth convex functions \citep{nemirovsky1983problem, nesterov2004introductory}. The iterative structure of FISTA also admits unrolling into a finite-depth differentiable computation \citep{Gregor2010LISTA}, which we exploit in Section~\ref{sec:method-unrolled} to train the encoder end-to-end through the solver.

\subsection{Preliminaries}
Here we summarize the notations used in this paper.
Let $[n] = \{1, \cdots, n\}$. 
Let $\mathbb{S}^{d-1}$ denote the $d$-dimensional unit sphere. For a set $S\subseteq [n]$ and matrix $X \in \mathbb{R}^{m \times n}$, $X_S \in \mathbb{R}^{m \times |S|}$ denotes the submatrix of $X$ with indices in $S$. Whenever it is apparent that $S$ is a subset of $[n]$, we use the notation $S^c$ to denote the complementary set $[n] \backslash S$. 
Also, for some vector $x \in \mathbb{R}^n$, define its support as $ \mathrm{supp}(x) = \{ \ i \ | \ i \in [n], \ x_i \neq 0 \}. $ Finally, $[ \ \cdot \ ]_+ = \max(\ \cdot \ , 0)$ denotes the ReLU function. 
\section{Theoretical Analysis}
\label{sec:theory}

This section establishes the central theoretical guarantee of NNN decoding: on any corpus, NNN decoding can handle a set of queries that is larger than or equal to the set solvable by dense retrieval, with the inequality holding for corpora with correlated documents. We first define the two retrievers and the sets of queries each one handles correctly, then state the main theorem along with a companion proposition establishing the inequality. 

\subsection{Retrieval methods and their success sets}
\label{sec:decoders}
Throughout this section we assume that the embeddings are given in dimension $d$. 
Let $\{u_i \in \mathbb{S}^{d-1}\}_{i \in [N]}$ be the embeddings of the $N$ documents, which constitute the corpus embedding matrix $U = [u_1, \ldots, u_N] \in \mathbb{R}^{d \times N}$ with unit columns. 
We assume that for a query, only a subset of documents $S\subseteq [N]$ with $|S| = k$ are relevant inside the corpus, and the retrieval task is to recover $S$ from only $U$ and the embedding of the query $v\in \mathbb{R}^d$. 
A retriever is a map that takes $(U, v)$ and returns a subset of $[N]$; we call the \emph{success set} of a retriever given $U, S$ as the set of $v$ which \emph{succeeds} in recovering $S$. 

\paragraph{Dense retrieval.} The standard bi-encoder retrieval pipeline employs the \emph{dense retriever}, which scores each document by its inner product with the query and returns the top $k$ by score. Formally, the dense retriever succeeds on $(U, v, S)$ when $v$ is a query which induce $S$ as the top-$k$ with positive scores. The success set is thus defined by 
\[
  \Phi_{\mathrm{DR}}(U, S) \;=\; \Bigl\{v \in \mathbb{S}^{d-1} : \max_{j \in S^c} u_j^\top v < \min_{i \in S} u_i^\top v \quad \text{and} \quad  \min_{i \in S} u_i^\top v > 0 \Bigr\}.
\]

This condition is \emph{local}: each document is scored independently using only its similarity to the query, without considering other documents
Consequently, documents with similar embeddings obtain nearly identical scores, even if one already captures the relevant information and the other is redundant, lacking a mechanism to suppress unnecessary duplicates. This motivates a decoder that scores documents jointly.

\paragraph{Non-negative elastic net decoder.} The NNN decoder is a retriever which returns the support of the non-negative elastic net minimizer,
\begin{equation}\label{eq:enet-theory}
   w^\star(v, \lambda_1, \lambda_2) \;=\; \argmin_{w \ge 0}\; \Bigg( \frac{1}{2}\|Uw - v\|_2^2 \;+\; \lambda_1 \| w \|_1 \;+\; \frac{\lambda_2}{2}\|w\|_2^2 \Bigg).
\end{equation}
where $\lambda_1, \lambda_2 \ge 0$ are hyperparameters. The success set of this retriever is the set of queries whose elastic net support can be made equal to $S$ by some choice of hyperparameters,
\[
  \Phi_{\mathrm{NNN}}(U, S) \;=\; \Bigl\{v \in \mathbb{S}^{d-1} : \exists \, \lambda_1, \lambda_2 \ge 0 \quad  \text{s.t} \quad \mathrm{supp}(w^\star(v, \lambda_1, \lambda_2)) = S \Bigr\}.
\]
Note that a document's inclusion in the support is determined in the context of the entire corpus, as the coefficient assigned to each $u_i$ depends on every other $u_j$ through the reconstruction term $\|Uw - v\|_2^2$. This is a crucial difference from the dense encoder, where scores are assigned to each document in isolation. 

\paragraph{What the success sets describe.} 

Each set describes which queries lead to the correct retrieval of $S$ under different scoring methods. At fixed $U$, a query $v \in \Phi_{\mathrm{NNN}}(U, S) \setminus \Phi_{\mathrm{DR}}(U, S)$ is one for which the same corpus embeddings are sufficient to identify $S$ under NNN decoding but not under dense retrieval. Our main result characterizes this difference set $ \Phi_{\mathrm{NNN}}(U, S) \setminus \Phi_{\mathrm{DR}}(U, S)$.

\subsection{Main result}
\label{sec:main-result}

We are ready to state the theoretical result of the paper, whose full proof is postponed to the Appendix. 

\begin{restatable}{theorem}{theoremone}
\label{thm:inclusion}
$\Phi_{\mathrm{DR}}(U,S)\subseteq\Phi_{\mathrm{NNN}}(U,S)$; For any corpus $U$ and target subset $S$, every query that dense retrieval handles correctly is also handled correctly by NNN decoding.
\end{restatable}

\begin{restatable}{proposition}{propositiontwo}
\label{prop:strict}
    There exist $U$ and $S$ for which
$\Phi_{\mathrm{NNN}}(U,S)\not\subseteq\Phi_{\mathrm{DR}}(U,S)$.
\end{restatable}

\begin{proof}[Proof outline of Theorem~\ref{thm:inclusion}]
Take $v \in \Phi_{\rm DR}(U, S).$ The KKT conditions for \eqref{eq:enet-theory} are  
\begin{align}
  w_i^\star > 0 \;&\Longrightarrow\; u_i^\top(v - U w^\star) = \lambda_1 + \lambda_2 w_i^\star, 
    \tag{A}\label{eq:kkt-A}\\
  w_i^\star = 0 \;&\Longrightarrow\; u_i^\top(v - U w^\star) \le \lambda_1.
    \tag{B}\label{eq:kkt-B}
\end{align}
Following the primal-dual witness construction \citep{Wainwright2009Lasso}, we construct a candidate solution $w^\star$ only supported on $S$, and prove that it satisfies \eqref{eq:kkt-A} and \eqref{eq:kkt-B}. These are rewritten as 
\begin{align}
    &w_S^\star =  (U_S^\top U_S + \lambda_2 I)^{-1}(U_S^\top v - \lambda_1 \mathbf{1}_k) \ , \label{eq:kkt-active_} \\
    \forall i \in S^c, \qquad & u_i^\top (v - U_S w_S^\star) \leq \lambda_1 \label{eq:kkt-inactive_}.
\end{align}
Now, define $\alpha = \min_{i \in S} u_i^\top v$, $\beta = \max_{j \in S^c} u_j^\top v$, 
and $\delta = (\alpha - \max(0,\beta))/2$. One can show that the choice 
$\lambda_1 = (\alpha + \max(0,\beta))/2$ and $\lambda_2 = k/\delta + k - 1$ 
suffices to certify \eqref{eq:kkt-active_} and \eqref{eq:kkt-inactive_}. In addition, strong convexity from $\lambda_2 > 0$ verifies $w^\star$ as the unique minimizer of \eqref{eq:enet-theory}, which is supported on $S$ by definition, thereby completing the proof. 
\end{proof}
\begin{proof}[Proof outline of Proposition~\ref{prop:strict}]
This can be verified by construction; set $d = N = 3$ and 
\[
  u_1=e_1,\quad u_2=\tfrac{1}{\sqrt{2}}(e_1+e_2),\quad u_3=e_3,\quad
  v=\tfrac{2}{3}e_1+\tfrac{2}{3}e_2+\tfrac{1}{3}e_3,\quad S=\{2,3\}.
\]
For $\lambda_1 \in (0, 1/3)$ and (say) $\lambda_2 = 0$ we have that $w_2^\star = 2\sqrt{2}/3 - \lambda_1$ and $w_3^\star = 1/3 - \lambda_1$ with both being positive. Hence \eqref{eq:kkt-active_} holds. 
Now, condition \eqref{eq:kkt-inactive_} reduces to 
\begin{equation}\label{eq:demo}
      u_1^\top(v - U_S w_S^\star) 
  \;=\; \tfrac{2}{3} - \tfrac{1}{\sqrt{2}}\bigl(\tfrac{2\sqrt{2}}{3} - \lambda_1\bigr) 
  \;=\; \tfrac{\lambda_1}{\sqrt{2}} \;<\; \lambda_1,
\end{equation}
So $v \in \Phi_{\rm NNN}(U, S)$, yet $u_2^\top v > u_1^\top v > u_3^\top v$. 
 \end{proof}

Together, Theorem~\ref{thm:inclusion} and Proposition~\ref{prop:strict} establish that at any fixed embedding dimension $d$ and for any corpus $U$,  the NNN decoder correctly recovers $S$ on at least as many queries as dense retrieval, and on strictly more queries for some $U$.

\subsection{Discussion of the result}
\label{sec:theory-discussion}

Theorem~\ref{thm:inclusion} and Proposition~\ref{prop:strict} together identify a structural limitation of inner-product scoring and a concrete mechanism by which joint decoding overcomes it. Two consequences of the result deserve emphasis, each connecting the theory to the methods and experiments that follow. 
\paragraph{A mechanism for the strict gap.} 

The KKT conditions in the proof sketch make Proposition~\ref{prop:strict} concrete. 
As apparent from \eqref{eq:demo}, \eqref{eq:kkt-inactive_} holds at $j = 1$ even though $u_1^\top v$ exceeds $u_3^\top v$ . 
This is because the relevant element $U_S w_S^\star$ absorbs the component of $v$ along $u_1$, shrinking the residual 
that $u_1$ sees below the $\ell_1$ threshold $\lambda_1$.

In a more general note, condition \eqref{eq:kkt-inactive_} screens an irrelevant item $u_j$ by the inner product of $u_j$ and the \emph{residual} $v - U_S w_S^\star$. This is in contrast to inner-product scoring, where the screening is done by the inner product of $u_j$ and $v$ alone. 
The term $u_j^\top U_S w_S^\star$ in \eqref{eq:kkt-inactive_} subtracts some portion of $u_j$ lying in 
$\mathrm{span}(U_S)$ from $u_j ^\top v$, so correlation between $u_j$ and $S$ aids suppression under NNN decoding.
The performance gap is therefore expected to be largest on corpora where 
relevant and irrelevant atoms are correlated, which is typical with corpora including several semantically related documents.

\paragraph{Towards a practical algorithm.} Theorem~\ref{thm:inclusion} shows that at any fixed $U$, NNN decoding has strictly more per-query expressive power than inner-product scoring. This has two
practical consequences. \textbf{First, NNN decoding is a principled drop-in replacement for inner-product scoring on the frozen embeddings of any pretrained bi-encoder}; while this requires a specific pair of $(\lambda_1,\lambda_2)$ for each query in theory, we discuss a practical way to avoid this in Subsection~\ref{sec:method-inference}.

\textbf{
Second, since the success set $\Phi_{\rm NNN}(U, S)$ depends on $U$, the bi-encoder can be trained to shape $U$ and $v$ so that more queries of interest fall inside it. 
} Concretely, fine-tuning the bi-encoder against a loss tailored to the NNN decoder can potentially build gains beyond applying NNN decoding to frozen embeddings, which is discussed in Subsection~\ref{sec:method-unrolled}

\section{Elastic net decoding: inference and training}
\label{sec:elastic-net-decoding}
In this section, the details of NNN decoding, along with the training method for the corpus and query embeddings are explained. Here, we assume that each query embedding $v$ and corpus embeddings $u_i, \ i \in [N]$ are produced via encoders with trainable parameters $\theta$ and $\phi$, denoted as $E_\theta$ and $E_\phi$ : 
\[
v = E_\theta (q) , \qquad u_i = E_\phi(d_i),
\]
where $q$ and $\{d_i\}_{i \in [N]}$ are the text representations of the query and corpus respectively.  

\subsection{Inference with NNN decoding}
\label{sec:method-inference}

Given query embedding $v$ and precomputed corpus matrix $U$, NNN decoding solves the non-negative elastic net
\eqref{eq:enet-theory} and returns the support $\mathcal{S}(v) := \mathrm{supp}(w^\star)$ as the retrieved set, ranked internally by the nonzero coefficients. For $\lambda_2 > 0$, the objective is strongly convex and $w^\star$ is unique. Therefore, this can be solved using the FISTA \citep{beck2009fista} algorithm, whose non-negative variant iterates
\begin{equation}\label{eq:fista-update}
\begin{aligned}
      w^{(t+1)} &= \Bigl[\big(1 - \tfrac{\lambda_2}{L}\bigr) z^{(t)}
              + \tfrac{1}{L} U^\top(v - U z^{(t)}) - \tfrac{\lambda_1}{L}\Bigr]_+,\\
  z^{(t+1)} &= w^{(t+1)} + \beta^{(t)}\bigl(w^{(t+1)} - w^{(t)}\bigr),
\end{aligned}
\end{equation}
with precomputed Lipschitz constant $L = \|U^\top U\|_2 + \lambda_2$ and $\beta^{(t)} = (\tau^{(t)} - 1)/\tau^{(t+1)}$ for
$\tau^{(t+1)} = \big(1 + \sqrt{1 + 4(\tau^{(t)})^2}\big) / 2$. We take $w^\star \approx w^{(T)}$ for a fixed number of iterations $T$.

\paragraph{Inference time cost.} Each FISTA operation is dominated by the matrix-vector products $U z^{(t)}$ and 
$U^\top(\cdot)$ which totals to a computational complexity of $O(dNT)$ under $T$ iterations. 
Therefore, the cost remains linear in both the embedding dimension $d$ and corpus size $N$ for constant $T$, inheriting the same scalability as inner-product scoring.

\paragraph{Hyperparameter choice.} As discussed in Section~\ref{sec:theory-discussion}, Theorem~\ref{thm:inclusion} is a \emph{per-query} guarantee on the existence of \emph{some}
$(\lambda_1, \lambda_2)$ under which NNN decoding succeeds, and does not promise that a \emph{single} pair shared across queries captures this gap. Thus, the optimal deployment of NNN decoding would select $(\lambda_1, \lambda_2)$ per query, which is difficult in practice. We instead adopt the simpler method of fixing a single $(\lambda_1, \lambda_2)$ across all queries, tuned by grid search on a held-out validation set. 
Although this is different from the theoretical setting in Section~\ref{sec:theory}, we find empirically (Section~\ref{sec:experiments}) that this choice already yields consistent gains over naive inner-product scoring even with the same embeddings.

\subsection{Unrolled FISTA for end-to-end training}
\label{sec:method-unrolled}
Up to this point, NNN decoding is assumed to be applied to predefined query and corpus embeddings generated from a fixed encoder model. In this subsection, we describe how to fine-tune the encoder model such that it generates embeddings which are more suitable for NNN decoding. 

Every operation in \eqref{eq:fista-update} is differentiable almost everywhere. Truncating FISTA to $T$ iterations therefore yields a finite-depth differentiable map $\Psi_T : (U, v) \mapsto w^{(T)}$, whose Jacobian is obtained by ordinary backpropagation through the unrolled computational graph \citep{Gregor2010LISTA}. This is exploited to train the bi-encoder end-to-end against a loss defined on the solver's output, or equivalently the NNN decoder's retrieval performance. 
In particular, for a training dataset $\mathcal{D}$ with respect to a corpus $\{d_i\}_{i \in [N]}$, wherein each training sample consists of a pair of queries $q$ and its relevant set $S$, we propose to minimize the following loss function: 
\begin{equation}\label{eq:training-objective}
  \mathcal{L}\big(\theta, \phi  \ | \  \mathcal{D}, \{d_i\}_{i \in [N]} \big) \;=\; \sum_{(q, S) \in \mathcal{D}} \Bigl[\, \gamma\,\tau\log\ \sum_{j\in S^c}\!e^{w^{(T)}_j/\tau}
      \;+\; \tau\log \ \sum_{i\in S} e^{-w^{(T)}_i/\tau}\,\Bigr]_+, 
\end{equation}
which is a relaxation of $\bigl[\gamma\,\max_{j\in S^c} w_j - \min_{i\in S} w_i\bigr]_+$ via temperature parameter $\tau$ with margin factor $\gamma > 1$. The softmin over $S$ and the softmax over $S^c$ enforce \emph{completeness}, in which every relevant coefficient is above every irrelevant one by a multiplicative margin. 
Note that the dependency on $\theta, \phi$ for the right-hand side of the loss function occurs via $w^{(T)} = \Psi_T\big( [ E_\phi(d_i) ]_{i \in [N]} , E_\theta(q) \big).$ In practice, we freeze the original corpus embedding model and propose to train a 2-layer MLP on top of it, taking its parameters as $\phi$. 
The hyperparameters $(\lambda_1, \lambda_2)$, chosen in advance, parameterize the solver used on both the forward pass of training and inference, and are fixed throughout training. 

\section{Experiments}
\label{sec:experiments}

\subsection{Experimental Setup}
We evaluate NNN decoding from two axes. The first axis freezes the embeddings specifically fine-tuned for inner-product scoring and varies only the decoder, isolating gains from replacing inner-product scoring with NNN decoding. The second axis additionally fine-tunes the encoder through the Unrolled FISTA solver, capturing the gains from learning embeddings specifically suited to NNN decoding. Full reproducibility details are deferred to Appendix~\ref{sec:reproducibility}. 

\paragraph{Datasets.} We evaluate on five datasets benchmarking retrieval tasks ranging from tool selection to multi-hop reasoning. 
\textbf{ToolBank} \citep{moon2024toolbank}, which consist of the \textbf{NumpyBank, AWSBank}, and \textbf{PandasBank} datasets, benchmarks the performance of a retrieval system on selecting all necessary tools to complete a task query within the domain of API-programming. \textbf{ToolLens} \citep{Qu2024CompletenessOriented} is another tool retrieval benchmark for more ubiquitous tasks. In both cases, each document in the corpus corresponds to a detailed description of each tool. Finally, \textbf{MultiHop-RAG} \citep{tang2024multihoprag} is a dataset benchmarking the retrieval system's performance on selecting multiple complementary passages to answer a query. 
All datasets consist of a set of queries, each accompanied by a set of ground-truth documents.  

\paragraph{Evaluation.} For each benchmark, we report Recall@$k$ for $k = 3,5$, which credits the partial recovery of the relevant set $S$. 
We additionally report Completeness@$k$ (Comp@$k$), defined as the fraction of queries for which \emph{every} relevant document appears in the top-$k$. 
The NNN decoder is evaluated by ranking the corpus in the support of $w^{(T)}$ by coefficient magnitude. 

\paragraph{Baseline and Implementation Details.} All methods share a backbone bi-encoder \texttt{bge-small-en-v1.5} \footnote{\url{https://huggingface.co/BAAI/bge-small-en-v1.5}}; additional experiments using different backbones are given in Appendix~\ref{sec:additional_exp}. The baseline and NNN decoding methods are configured as the following: 
\begin{itemize}[leftmargin=0.5cm] 
    \item \textsc{\textbf{Dense}} is a dense retriever with the backbone fine-tuned using InfoNCE \citep{oord2018infonce, karpukhin2020dpr} on each dataset. 
    \item \textsc{\textbf{MMR}} (Maximal Marginal Relevance) \citep{Carbonell1998MMR} is a greedy re-ranking method for diverse retrieval which, at each iteration, selects a document $i$ via $ \argmax_{i \in [N] \backslash R} ( \lambda u_i^\top v - (1-\lambda) \max_{j \in R} u_i^\top u_j)$, where $R$ is the set of currently chosen documents, and $\lambda \in [0,1]$ is a hyperparameter controlling the diversity over the retrieved results. \textsc{Dense} is used for inner-product scoring. 
    \item \textsc{\textbf{COLT}} \citep{Qu2024CompletenessOriented} is a collaborative learning-based method to further fine-tune bi-encoders for selecting all complementary tools necessary to complete a task. As is done in \citep{Qu2024CompletenessOriented}, we use \textsc{dense} as the initial checkpoint when performing \textsc{COLT}. 
    Note that this only changes the training procedure, and still utilizes inner-product scoring for retrieval. 
    \item \textsc{\textbf{NNN-Fix}} applies the NNN decoder \eqref{eq:enet-theory} on the frozen embeddings of \textsc{dense}, with a single pair of $(\lambda_1, \lambda_2)\in [0.01, 0.03, 0.06, 0.1, 0.3, 0.6, 1.0]^2$ via validation on a held-out dataset. 
    As an ablation study, we also consider the case when only the $\ell_1$ or $\ell_2$ penalty term is in the decoder, referring to it as \textsc{\textbf{L1-Fix}} and \textsc{\textbf{L2-Fix}}, respectively. 
    \item \textsc{\textbf{NNN-Tr}} initializes from the \textsc{dense} checkpoint and continues training with the ranking loss \eqref{eq:training-objective} back-propagated through $T=50$ unrolled FISTA iterations, then retrieves with NNN decoding. The hyperparameters $(\lambda_1, \lambda_2)$, shared in both training and decoding, are selected likewise to \textsc{NNN-Fix}. During the decoding stage, we vary the FISTA iterations to evaluate its performance against inference cost. \textsc{\textbf{L1-Tr}} and \textsc{\textbf{L2-Tr}} variants are defined likewise. 
\end{itemize}
 Note that fine-tuning on \textsc{{dense}} was performed until early-stopping, to ensure that the dense encoder is fully trained, and therefore any performance gain observed on \textsc{NNN-Tr} and its variants are not a result of simply a higher number of epochs. 

\subsection{Experimental Results}

\begin{table*}[t]
\centering
\small
\setlength{\tabcolsep}{3.3pt}
\caption{Test-set Recall@3 and Recall@5 (\%) per dataset. Best per column is \textbf{bold}. Subscripts in {\color{gainpos}blue} and {\color{gainneg}red} denote performance increase and decrease from \textsc{Dense}, respectively. }
\begin{tabular*}{\textwidth}{@{\extracolsep{\fill}}lllllllllll@{}}
\toprule
Method & \multicolumn{2}{c}{NumpyBank} & \multicolumn{2}{c}{PandasBank} & \multicolumn{2}{c}{AWSBank} & \multicolumn{2}{c}{ToolLens} & \multicolumn{2}{c}{MultiHop-RAG} \\
\cmidrule(lr){2-3}\cmidrule(lr){4-5}\cmidrule(lr){6-7}\cmidrule(lr){8-9}\cmidrule(lr){10-11}
 & R@3 & R@5 & R@3 & R@5 & R@3 & R@5 & R@3 & R@5 & R@3 & R@5 \\
\midrule
\textsc{dense} & 66.9 & 79.5 & 40.1 & 49.7 & 63.4 & 75.3 & 81.2 & 92.5 & 77.4 & 88.4 \\
\textsc{MMR} & \cg{71.1}{+4.1} & \cg{80.6}{+1.1} & \cg{42.8}{+2.7} & \cg{51.2}{+1.4} & \cg{68.5}{+5.1} & \cg{77.6}{+2.3} & \cg{88.8}{+7.6} & \cg{97.1}{+4.6} &\cg{76.7}{-0.8} & \cg{88.3}{-0.1} \\
\textsc{COLT}
 & \cg{67.9}{+1.0} & \cg{80.2}{+0.7}
 & \cg{39.7}{-0.4} & \cg{49.0}{-0.7}
 & \cg{63.2}{-0.2} & \cg{75.1}{-0.2}
 & \cg{89.1}{+7.9} & \cg{97.0}{+4.5}
 & \cg{80.1}{+2.7} & \cg{89.1}{+0.7} \\
\midrule
\textsc{NNN-Fix}
 & \cg{73.6}{+6.7} & \cg{84.2}{+4.7}
 & \cg{46.8}{+6.7} & \cg{55.7}{+6.0}
 & \cg{72.0}{+8.6} & \cg{80.6}{+5.3}
 & \cg{88.0}{+6.8} & \cg{96.1}{+3.6}
 & \cg{77.3}{-0.1} & \cg{89.8}{+1.4} \\
\textsc{\hspace*{1.1em}L1-Fix}
 & \cg{73.5}{+6.6} & \cg{83.9}{+4.4}
 & \cg{46.9}{+6.8} & \cg{55.6}{+5.9}
 & \cg{71.9}{+8.5} & \cg{80.6}{+5.3}
 & \cg{88.5}{+7.3} & \cg{95.6}{+3.1}
 & \cg{72.7}{-4.7} & \cg{82.2}{-6.2} \\
\textsc{\hspace*{1.1em}L2-Fix}
 & \cg{71.9}{+5.0} & \cg{82.8}{+3.3}
 & \cg{40.4}{+0.3} & \cg{50.0}{+0.3}
 & \cg{65.3}{+1.9} & \cg{76.8}{+1.5}
 & \cg{87.5}{+6.3} & \cg{96.0}{+3.5}
 & \cg{76.4}{-1.0} & \cg{89.3}{+0.9} \\
\midrule
\textsc{NNN-Tr}
 & \cg{74.7}{+7.8}        & \cg{\textbf{85.2}}{+5.7}
 & \cg{\textbf{49.9}}{+9.8}& \cg{\textbf{58.1}}{+8.4}
 & \cg{73.4}{+10.0}       & \cg{\textbf{81.6}}{+6.3}
 & \cg{96.1}{+14.9}       & \cg{\textbf{98.4}}{+5.9}
 & \cg{84.4}{+7.0}        & \cg{\textbf{91.9}}{+3.5} \\
\textsc{\hspace*{0.85em}L1-Tr}
 & \cg{74.8}{+7.9}        & \cg{85.0}{+5.5}
 & \cg{48.8}{+8.7}        & \cg{57.4}{+7.7}
 & \cg{\textbf{73.5}}{+10.1}& \cg{\textbf{81.6}}{+6.3}
 & \cg{96.2}{+15.0}       & \cg{98.2}{+5.7}
 & \cg{\textbf{86.1}}{+8.7}& \cg{91.5}{+3.1} \\
\textsc{\hspace*{0.85em}L2-Tr}
 & \cg{\textbf{74.9}}{+8.0}& \cg{85.0}{+5.5}
 & \cg{48.5}{+8.4}        & \cg{57.6}{+7.9}
 & \cg{73.3}{+9.9}        & \cg{81.4}{+6.1}
 & \cg{\textbf{96.3}}{+15.1}& \cg{\textbf{98.4}}{+5.9}
 & \cg{82.3}{+4.9}        & \cg{89.9}{+1.5} \\
\bottomrule
\end{tabular*}

\vspace*{6pt}
\label{tab:recall-results}
\centering
\small
\setlength{\tabcolsep}{3.3pt}
\caption{Test-set Comp@3 and Comp@5 (\%) per dataset. Best per column is \textbf{bold}.  Subscripts in {\color{gainpos}blue} and {\color{gainneg}red} denote performance increase and decrease from \textsc{Dense}, respectively.  }
\begin{tabular*}{\textwidth}{@{\extracolsep{\fill}}lllllllllll@{}}
\toprule
Method & \multicolumn{2}{c}{NumpyBank} & \multicolumn{2}{c}{PandasBank} & \multicolumn{2}{c}{AWSBank} & \multicolumn{2}{c}{ToolLens} & \multicolumn{2}{c}{MultiHop-RAG} \\
\cmidrule(lr){2-3}\cmidrule(lr){4-5}\cmidrule(lr){6-7}\cmidrule(lr){8-9}\cmidrule(lr){10-11}
 & C@3 & C@5 & C@3 & C@5 & C@3 & C@5 & C@3 & C@5 & C@3 & C@5 \\
\midrule
\textsc{dense} & 30.8 & 50.9 & \hspace{0.25em}7.3 & 14.1 & 30.0 & 45.1 & 55.2 & 81.5 & 50.7 & 75.1 \\
\textsc{MMR}
 & \cg{38.6}{+7.8} & \cg{54.7}{+3.7}
 & \cg{9.9}{+2.5}  & \cg{15.8}{+1.7}
 & \cg{36.4}{+6.4} & \cg{49.7}{+4.6}
 & \cg{72.2}{+17.0}& \cg{93.1}{+11.6}
 & \cg{48.9}{-1.8} & \cg{74.7}{-0.4} \\
\textsc{COLT}
 & \cg{32.9}{+2.1} & \cg{52.5}{+1.6}
 & \cg{7.2}{-0.1}  & \cg{13.5}{-0.6}
 & \cg{29.9}{-0.1} & \cg{45.0}{-0.1}
 & \cg{73.5}{+18.3}& \cg{93.6}{+12.1}
 & \cg{54.2}{+3.5} & \cg{76.4}{+1.3} \\
\midrule
\textsc{NNN-Fix}
 & \cg{41.0}{+10.2}& \cg{61.4}{+10.5}
 & \cg{12.6}{+5.3} & \cg{20.3}{+6.2}
 & \cg{39.4}{+9.4} & \cg{55.3}{+10.2}
 & \cg{72.0}{+16.8}& \cg{91.4}{+9.9}
 & \cg{50.2}{-0.5} & \cg{76.4}{+1.3} \\
\textsc{\hspace{0.85em} L1-Fix}
 & \cg{41.3}{+10.5}& \cg{61.2}{+10.3}
 & \cg{12.4}{+5.1} & \cg{20.4}{+6.3}
 & \cg{39.4}{+9.4} & \cg{55.6}{+10.5}
 & \cg{73.4}{+18.2}& \cg{90.0}{+8.5}
 & \cg{40.0}{-10.7}& \cg{56.4}{-18.7} \\
\textsc{\hspace{0.85em} L2-Fix}
 & \cg{37.5}{+6.7} & \cg{57.9}{+7.0}
 & \cg{6.4}{-0.9}  & \cg{12.7}{-1.4}
 & \cg{30.1}{+0.1} & \cg{46.6}{+1.5}
 & \cg{70.5}{+15.3}& \cg{91.7}{+10.2}
 & \cg{45.3}{-5.4} & \cg{74.7}{-0.4} \\
\midrule
\textsc{NNN-Tr}
 & \cg{42.6}{+11.8}        & \cg{\textbf{63.1}}{+12.2}
 & \cg{\textbf{15.1}}{+7.8}& \cg{\textbf{23.1}}{+9.0}
 & \cg{41.3}{+11.3}        & \cg{\textbf{57.4}}{+12.3}
 & \cg{91.8}{+36.6}        & \cg{\textbf{97.0}}{+15.5}
 & \cg{63.1}{+12.4}        & \cg{\textbf{81.8}}{+6.7} \\
\textsc{\hspace{0.85em} L1-Tr}
 & \cg{\textbf{42.9}}{+12.1}& \cg{\textbf{63.1}}{+12.2}
 & \cg{13.6}{+6.3}          & \cg{22.4}{+8.3}
 & \cg{\textbf{41.4}}{+11.4}& \cg{\textbf{57.4}}{+12.3}
 & \cg{92.1}{+36.9}         & \cg{96.3}{+14.8}
 & \cg{\textbf{67.6}}{+16.9}& \cg{\textbf{81.8}}{+6.7} \\
\textsc{\hspace{0.85em} L2-Tr}
 & \cg{42.6}{+11.8}        & \cg{62.7}{+11.8}
 & \cg{13.4}{+6.1}         & \cg{22.3}{+8.2}
 & \cg{41.2}{+11.2}        & \cg{56.7}{+11.6}
 & \cg{\textbf{92.4}}{+37.2}& \cg{\textbf{97.0}}{+15.5}
 & \cg{59.6}{+8.9}         & \cg{80.0}{+4.9} \\
\bottomrule
\end{tabular*}
\vspace*{-12pt}
\label{tab:comp-results}
\end{table*}

Tables~\ref{tab:recall-results} and ~\ref{tab:comp-results} compare the performance of our NNN decoding variants with the baseline \textsc{dense} and \textsc{COLT}. 
Immediately, one can confirm that directly performing NNN decoding on the frozen embeddings (\textsc{NNN-Fix}) yield performance gains on the ToolBank datasets over \textsc{COLT} and \textsc{MMR}, and additionally on ToolLens when compared with \textsc{dense}, highlighting its effectiveness as a drop-in replacement over dense retrieval. Performing end-to-end training (\textsc{NNN-Tr}) yield further gains across all datasets and metrics. This is prominent in the Comp@$3$ metric, where a 36.6\% performance gain in the ToolLens dataset over \textsc{dense} and a 18.3\% performance gain over \textsc{COLT} is observed.

\begin{wrapfigure}[11]{r}{0.4\textwidth}
\centering
\vspace*{-12pt}
\includegraphics[width=0.4\textwidth]{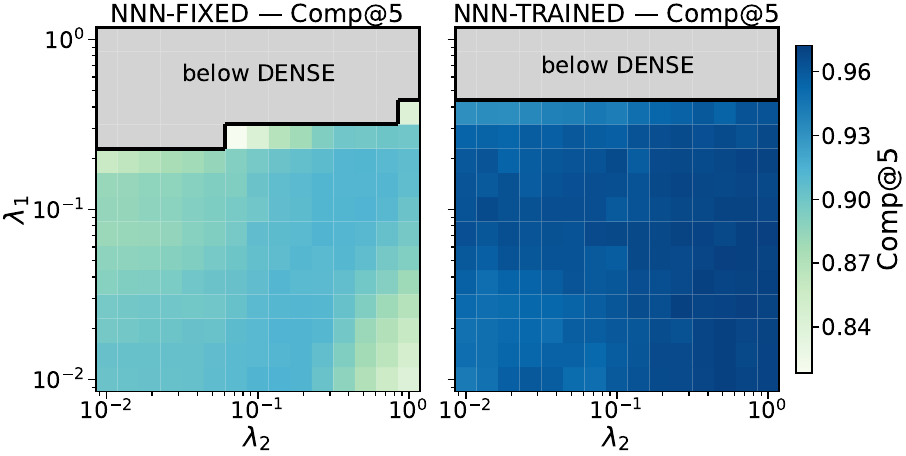}
\caption{ Comp@5 of \textsc{NNN-Fix} and \textsc{NNN-Tr} evaluated on a grid of $(\lambda_1, \lambda_2)$ for ToolLens. Compared to \textsc{NNN-Fix}, \textsc{NNN-Tr} is more robust. 
}
\label{fig:lambda_sweep}
\end{wrapfigure}
\paragraph{Hyperparameter sensitivity.} 
Tables~\ref{tab:recall-results} and~\ref{tab:comp-results} show that with frozen embeddings, squared $\ell_2$ regularization alone underperforms \textsc{NNN-Fix} substantially, while $\ell_1$ alone incurs a milder drop except on MultiHop-RAG. Once the embeddings are trained for NNN decoding, both variants are competitive. Consistent with this, Figure~\ref{fig:lambda_sweep} shows Comp@5 is more sensitive to $(\lambda_1, \lambda_2)$ under frozen embeddings. Thus, hyperparameter selection is crucial for frozen-embedding NNN decoding, but fine-tuning renders the decoder robust to it.

\paragraph{Performance against FISTA iterations.}
It is important to investigate how the number of FISTA iterations affect the performance of the NNN decoder in terms of inference time compute. Figure~\ref{fig:FISTA_iter} displays the change in performance of both \textsc{NNN-Fix} and \textsc{NNN-Tr} in terms of Comp@5. 
Here, we see that the behavior is dependent on the dataset we evaluate on, where \textsc{NNN-Fix} performance plateaus even at small iterations for MultiHop-RAG, while it continually increases for others. The same plateauing can also be observed for \textsc{NNN-Tr} for the ToolLens and MultiHop-RAG. 
Nevertheless, even small number of iterations suffice to improve over \textsc{dense}.

\paragraph{Performance against number of ground truth items.}
Figure~\ref{fig:GT_strat} plots Comp@5 of \textsc{dense}, \textsc{NNN-Fix}, and \textsc{NNN-Tr} conditioned on the number of ground truth items per query, $|S|$. When $|S|$ is small ($|S| \in \{1, 2\}$), the three methods are consistently similar, indicating that with few relevant documents, NNN decoding has little room to contribute. As $|S|$ grows, \textsc{dense} deteriorates sharply, while both NNN decoding variants degrade far more mildly. This is particularly pronounced in the ToolBank datasets.  
This phenomenon can be interpreted through the mechanism explained in Section~\ref{sec:theory-discussion} by the following. Recall that dense retrieval fails when an item in $S$ has high correlation with a item outside of $S$, but this is avoidable using NNN decoding by choosing the best (and possibly most relevant) item out of the two. Since the likelihood of $S$ including such a difficult-to-distinguish item increases with $|S|$, a wider gap with increasing $|S|$ is to be expected.

\begin{figure}
    \centering
    \includegraphics[width=1.0\linewidth]{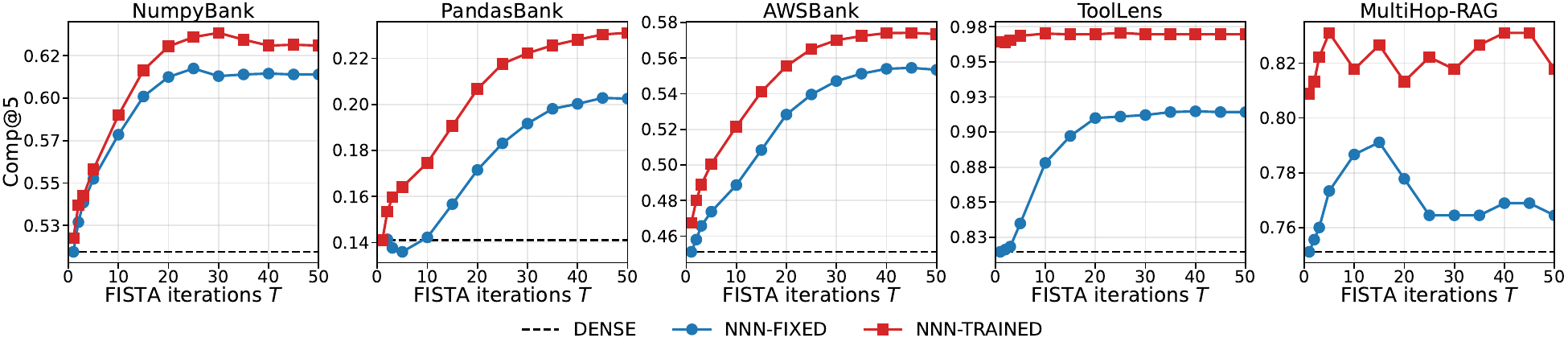}
    \caption{Comp@5 of \textsc{dense}, \textsc{NNN-Fix}, and \textsc{NNN-Tr} against FISTA iterations during inference.}
    \label{fig:FISTA_iter}
\vspace*{-6pt}
\end{figure}

\begin{figure}
    \centering
    \includegraphics[width=1.0\linewidth]{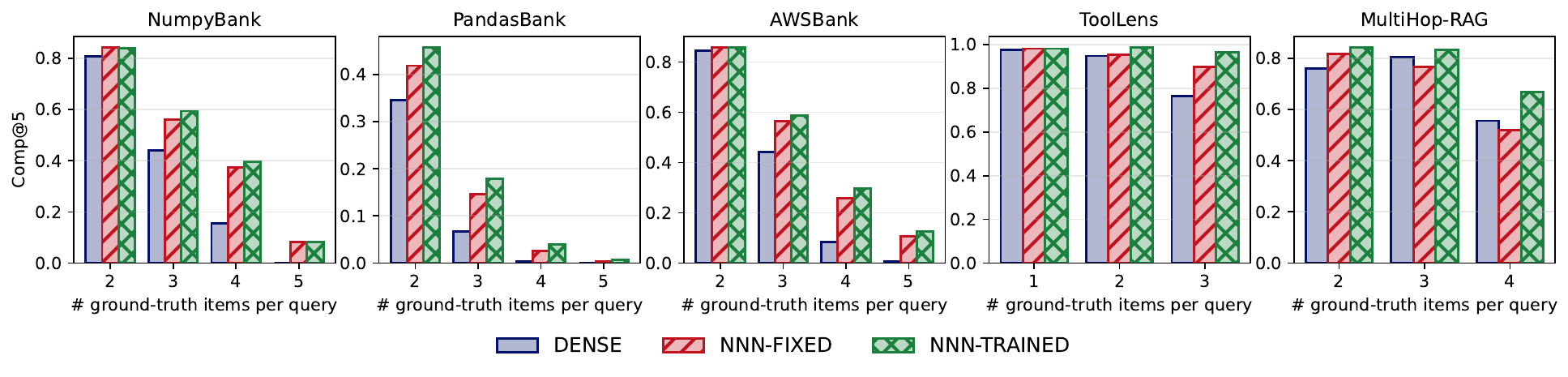}
    \caption{Comp@5 of \textsc{dense}, \textsc{NNN-Fix}, and \textsc{NNN-Tr} stratified over the number of ground truth items per query.}
    \label{fig:GT_strat}
\vspace*{-12pt}
\end{figure}
\section{Conclusion}
\label{sec:conclusion}
This paper presents Non-Negative elastic Net (NNN) decoding, a retrieval procedure that selects items by jointly reconstructing the query embedding as a sparse linear combination of their embeddings. The design is motivated by a structural limitation of independent inner product scoring which cannot resolve ambiguity among semantically similar documents, resulting in retrieving redundant information unfavorable for downstream tasks such as multi-hop inference and tool retrieval. 

Our central contribution is a comparison of these two retrieval procedures in both theory and practice. On the theoretical side, Theorem~\ref{thm:inclusion} shows that the queries correctly handled by inner-product scoring form a subset of those correctly handled by NNN decoding under appropriate regularization, and Proposition~\ref{prop:strict} exhibits corpora on which the containment is strict. With no assumptions on the encoder architecture, NNN decoding can be seen as a generalization of inner-product scoring which is robust against correlations among relevant and irrelevant items. 
On the empirical side, swapping inner-product scoring for NNN decoding at inference time, while keeping the bi-encoder left untouched, produces consistent gains on several benchmarks. This can be further improved by fine-tuning the encoder through an unrolled FISTA solver. 
We hope that our findings motivate a reassessment of inner-product scoring in dense retrieval, rather than considering it as a fixed component.

\paragraph{Limitations and future directions.} Our method does not fully realize the per-query nature of Theorem~\ref{thm:inclusion}; $(\lambda_1, \lambda_2)$ are selected globally on a held-out validation set, creating a mismatch between the theoretical guarantee and the practical deployment. 
 Developing a method selecting $(\lambda_1, \lambda_2)$ against each query from its embeddings would align deployment more closely with Theorem~\ref{thm:inclusion}, and likely result in further gains. 
In addition, alternative reconstruction methods beyond the non-negative elastic net merit study. In particular, solvers that avoid full passes over the corpus \citep{Fujiwara2016selectiveCD, Ida2023FastBCD, Ida2024FastIHTA} would be necessary for the inference cost to be comparable with fast dense retrieval using approximate nearest neighbor methods \citep{Malkov2020HNSW, guo2020scann, douze2024faiss}. 
On the training side, end-to-end fine-tuning through the unrolled solver incurs 
$O(dNT)$ memory per query, which limits \textsc{NNN-Tr} to moderate-scale corpora depending on available compute. A more efficient training procedure would be necessary to scale our method to larger corpora.

\bibliography{main}

\appendix

\section{Proofs}
\label{sec:proofs}
\textbf{Setup.} Let $U = [u_1,\ldots,u_N]\in\mathbb{R}^{d\times N}$ be a matrix of unit vectors and $S\subseteq[N]$ with $|S|=k\ge 1$. Define the correlation-gap region
\[
  \Phi_{\mathrm{DR}}(U, S) \;=\; \Bigl\{v \in \mathbb{S}^{d-1} : \max_{j \in S^c} u_j^\top v < \min_{i \in S} u_i^\top v \ \text{and} \;  \min_{i \in S} u_i^\top v > 0 \Bigr\}.
\]
and let $\Phi_{\mathrm{NNN}}(U,S)$ be the set of unit vectors $v$ for which there exist $\lambda_1,\lambda_2\ge 0$, not both zero, such that the unique minimizer 
\begin{equation}\label{eq:enet}
  w^\star(v,\lambda_1,\lambda_2)
  = \argmin_{w\ge 0} \Bigg(
    \frac{1}{2}\|Uw-v\|_2^2 + \lambda_1\mathbf{1}_N^\top w
    + \frac{\lambda_2}{2}\|w\|_2^2 \Bigg)
\end{equation}
has support exactly $S$. Note that $\|w\|_1 = \mathbf{1}_N^\top w$ since the minimization is conditioned on $w \ge 0$. 

\theoremone*

\begin{proof}
Fix $v$ with $\|v\|=1$, $v\in\Phi_{\mathrm{DR}}(U,S)$, and set
$\alpha:=\min_{i\in S}u_i^\top v>0$ and $\beta:=\max_{j\in S^c}u_j^\top v$.
Choose hyperparameters
\begin{equation}\label{eq:params}
  \lambda_1 = \frac{\alpha+\max(0,\beta)}{2},\qquad
  \delta =  \frac{\alpha - \max(0, \beta)}{2}, \qquad
  \lambda_2  = \frac{k}{\delta}+k-1.
\end{equation}

The KKT conditions for \eqref{eq:enet} require, for each coordinate $i$:
\begin{align}
  w_i^\star > 0 &\;\Longrightarrow\; u_i^\top v - u_i^\top Uw^\star - \lambda_1 = \lambda_2 w_i^\star, \label{eq:kkt-active}\\
  w_i^\star = 0 &\;\Longrightarrow\; u_i^\top v - u_i^\top Uw^\star \le \lambda_1. \label{eq:kkt-inactive}
\end{align}
Our proof strategy consists of constructing a candidate solution for our minimizer \eqref{eq:enet} which is only supported on $S$, and verifying that it satisfies the KKT conditions. 
Set $w_j^\star=0$ for $j\notin S$. Restricting \eqref{eq:kkt-active} to $S$ gives the linear system
\begin{equation}\label{eq:linsys}
  (U_S^\top U_S + \lambda_2 I)\,w_S^\star = U_S^\top v - \lambda_1\mathbf{1}_k,
\end{equation}
which has unique solution $w_S^\star=(U_S^\top U_S+\lambda_2 I)^{-1}(U_S^\top v-\lambda_1\mathbf{1}_k)$
since $U_S^\top U_S\succeq 0$ and $\lambda_2>0$.
It now remains to verify:
\begin{itemize}
  \item[(A)] $w_{S,i}^\star>0$ for all $i\in S$,
  \item[(B)] $u_j^\top v - u_j^\top U_S w_S^\star\le\lambda_1$ for all $j\in S^c$.
\end{itemize}

\smallskip\noindent\textit{Verification of (A).}
Write $U_S^\top U_S = I + G_S$. Since $G_S$ has zero diagonal and its elements have modulus no more than unity, $\|G_S\|_\infty\le k-1$ holds. Set
\[
  \rho := \frac{\|G_S\|_\infty}{1+\lambda_2} \le \frac{k-1}{1+\lambda_2} < 1,
\]
where the last inequality holds from our choice of $\lambda_2$.
Define the decoupled approximation
\begin{equation}\label{eq:wbar}
  \bar{w}_S := \frac{U_S^\top v - \lambda_1\mathbf{1}_k}{1+\lambda_2},
\end{equation}
so that \eqref{eq:linsys} reads $w_S^\star=(I+G_S/(1+\lambda_2))^{-1}\bar{w}_S$.
Since $\|A\|_\infty\le\rho<1$ for $A:=-G_S/(1+\lambda_2)$, the Neumann series and triangle inequality offer
\[
  \|w_S^\star-\bar{w}_S\|_\infty
  = \Bigl\|\sum_{n=1}^\infty A^n \bar{w}_S\Bigr\|_\infty
  \le \sum_{n = 1}^\infty \bigl\| A^n \bigl\|_\infty \|  \bar{w}_S \|_\infty
  \le \frac{\rho}{1-\rho}\|\bar{w}_S\|_\infty.
\]
Using the fact that 
\begin{equation*}
    \| \bar{w}_S \|_\infty = \frac{1}{1 + \lambda_2} \big( \max_{i \in S} u_i^\top v - \lambda_1 \big) \le \frac{1}{1+\lambda_2}, \quad \text{and} \quad
    \frac{\rho}{1 - \rho} \le \frac{k - 1}{\lambda_2 -k+2},
\end{equation*}
 we obtain
\begin{equation}\label{eq:err}
  \|w_S^\star-\bar{w}_S\|_\infty \;\le\; \frac{k-1}{(1+\lambda_2)(\lambda_2-k+2)}.
\end{equation}
Recall that by definition of $\alpha, \lambda_1$ and $\delta$, given in \eqref{eq:params},  
\begin{equation}\label{eq:wbar-lower}
  \bar{w}_{S,i} = \frac{u_i^\top v - \lambda_1}{1+\lambda_2} \ge \frac{1}{1 + \lambda_2} \big( \min_{i \in S} u_i^\top v - \lambda_1  \big) = \frac{\delta}{1 + \lambda_2}. 
\end{equation}
Combining \eqref{eq:err} and \eqref{eq:wbar-lower} offers 
\[
  w_{S,i}^\star \;\ge\; \bar{w}_{S,i} - \|w_S^\star-\bar{w}_S\|_\infty
  \ge \frac{1}{1+\lambda_2}\!\left(\delta - \frac{k-1}{\lambda_2-k+2}\right) \ge \frac{1}{1+\lambda_2}\!\left(\delta - \frac{k}{\lambda_2-k+2}\right) 
\]
for any $i \in S$. 
By \eqref{eq:params}, $\lambda_2-k+2=k/\delta+1$, so
\begin{equation}\label{eq:key}
  \frac{k}{\lambda_2-k+2} = \frac{k}{k+\delta} \ \delta < \delta,
\end{equation}
hence $w_{S,i}^\star>0$ for any $i \in S$, thus verifying condition (A).

\smallskip\noindent\textit{Verification of (B).}
From \eqref{eq:err}, $\|w_S^\star\|_\infty \le \|\bar{w}_S\|_\infty + \|w_S^\star-\bar{w}_S\|_\infty \le 1/(\lambda_2-k+2)$, so for all $j\in S^c$,
\begin{equation}\label{eq:inactive-bound}
  \bigl|u_j^\top U_S w_S^\star\bigr| \le k\|w_S^\star\|_\infty \le \frac{k}{\lambda_2-k+2} < \delta,
\end{equation}
where the last inequality follows from \eqref{eq:key}. 
Thus, 
\begin{equation}
     u_j^\top v - u_j^\top U_S w_S^\star \le \max_{j \in S^c} \Big( u_j^\top v  + \bigl|u_j^\top U_S w_S^\star\bigr| \Big) < \beta + \delta \le \lambda_1,
\end{equation}
which verifies condition (B).

The vector $w^\star$ satisfies conditions (A) and (B), hence the full KKT conditions \eqref{eq:enet}. Since the objective function is strongly convex from $\lambda_2>0$, this implies that $w^\star$ is the unique global minimizer. As $\mathrm{supp}(w^\star)=S$, we conclude $v\in\Phi_{\mathrm{NNN}}(U,S)$.
\end{proof}

\propositiontwo

\begin{proof}
Set $d=N=3$,
\[
  u_1=e_1,\quad u_2=\frac{1}{\sqrt{2}}(e_1+e_2),\quad u_3=e_3,\quad
  v=\frac{2}{3}e_1+\frac{2}{3}e_2+\frac{1}{3}e_3,\quad S=\{2,3\}.
\]
One computes
\begin{equation}\label{eq:ex-dots}
  u_1^\top v = \frac{2}{3}, \qquad
  u_2^\top v = \frac{2\sqrt{2}}{3}, \qquad
  u_3^\top v = \frac{1}{3},
\end{equation}
so $\min_{i\in S}u_i^\top v = \frac{1}{3} < \frac{2}{3} = \max_{j\in S}u_j^\top v$ and $v\notin\Phi_{\mathrm{DR}}(U,S)$.

We show $v\in\Phi_{\mathrm{NNN}}(U,S)$ with $\lambda_2=0$ and any $\lambda_1\in(0,\frac{1}{3})$. Since $U_S^\top U_S = I$, the system \eqref{eq:linsys} reduces to
\begin{equation}\label{eq:ex-sol}
  w_S^\star = U_S^\top v - \lambda_1\mathbf{1}
  = \begin{pmatrix}\frac{2\sqrt{2}}{3}-\lambda_1\\[4pt]\frac{1}{3}-\lambda_1\end{pmatrix} > 0,
\end{equation}
where positivity holds for $\lambda_1<\frac{1}{3}$. The inactive condition \eqref{eq:kkt-inactive} for $j=1$ gives
\begin{equation}\label{eq:ex-inactive}
  u_1^\top v - u_1^\top U_S w_S^\star
  = \frac{2}{3} - \frac{1}{\sqrt{2}}\!\left(\frac{2\sqrt{2}}{3}-\lambda_1\right)
  = \frac{\lambda_1}{\sqrt{2}} < \lambda_1,
\end{equation}
so the full KKT system is satisfied with $\mathrm{supp}(w^\star)=S$, giving $v\in\Phi_{\mathrm{NNN}}(U,S)$.
\end{proof}

\section{Reproducibility Information}
\label{sec:reproducibility}
\subsection{Datasets}
We evaluate on five benchmarks. The {NumpyBank},
{PandasBank}, and {AWSBank} datasets in ToolBank ship with their own
train / validation / test splits, which we use as released. {MultiHop-RAG}
and {ToolLens} do not provide validation splits, so they are constructed by the following: for {MultiHop-RAG}, we generate a random $80/10/10$
train / validation / test split, while for {ToolLens}, we retain the original
test split and randomly partition the original training set into
$80\%$ train and $20\%$ validation splits. The same splits are used throughout all experiments. The statistics of our dataset are summarized in Table \ref{tab:dataset-stats}. 

We note that the {ToolLens} repository also distributes a separate
``tuning'' set intended for hyperparameter selection. This is \emph{not} used,
since we verified that its query–tool pairs coincide with those of the official test set, and thus using it would amount to performing validation on the test data itself.

\begin{table}[b]
\centering
\caption{Dataset statistics. $|\mathcal{T}|$ is the corpus size. ``Tools/query''
reports the min / max / mean number of relevant items/documents per query on the
test split (statistics are nearly identical on train and validation). For
{NumpyBank}, {PandasBank}, and {AWSBank} we use the
released splits; for {MultiHop-RAG} we use a random $80/10/10$ split;
for {ToolLens} we keep the original test set and split the original
train set $80/20$ into train and validation.}
\label{tab:dataset-stats}
\begin{tabular}{lrrrrc}
\toprule
Dataset & $|\mathcal{T}|$ & $|Q_{\text{train}}|$ & $|Q_{\text{val}}|$
        & $|Q_{\text{test}}|$ & items/query \\
        & & & & & { (min / max / mean)} \\
\midrule
{NumpyBank}    &   511 & 15{,}994 & 1{,}998 & 2{,}000 & 2 / 6 / 2.88 \\
{PandasBank}   & 1{,}651 & 56{,}013 & 7{,}002 & 7{,}002 & 2 / 8 / 2.99 \\
{AWSBank}      & 1{,}002 & 58{,}227 & 7{,}278 & 7{,}278 & 2 / 6 / 3.04 \\
{ToolLens}     &   464 & 13{,}515 & 3{,}378 & 1{,}877 & 1 / 3 / 2.66 \\
{MultiHop-RAG} &   609 &  1{,}805 &    225 &    225 & 2 / 4 / 2.58 \\
\bottomrule
\end{tabular}
\end{table}

\subsection{\textsc{dense}}

\paragraph{Architecture.}
We fine-tune a dual bi-encoder consisting of separate query and corpus encoder models $E_\theta$ and $E_\phi$ respectively, both initialized from the same checkpoint and \emph{fully trainable} (no frozen parameters and no adapter).

\paragraph{InfoNCE Loss Function.}
Given a batch of $N$ (query, positive-corpus) pairs $\{(q_i, d_i)\}_{i\in [N]}$ and their embeddings $\{v_i = E_\theta(q_i), u_i = E_\phi(d_i)\}_{i \in [N]}$, we optimize the asymmetric in-batch InfoNCE loss \citep{karpukhin2020dpr}
\begin{equation}
    \mathcal{L}_{\text{InfoNCE}} = -\frac{1}{N} \sum_{i=1}^{N} \log \frac{\exp\!\big(\mathrm{sim}(v_i, u_i) / \tau\big)}{\sum_{j=1}^{N} \exp\!\big(\mathrm{sim}(v_i, u_j) / \tau\big)},
    \label{eq:infonce}
\end{equation}
where $\mathrm{sim}(\cdot,\cdot)$ is the cosine similarity. We use temperature $\tau = 0.1$. Each (query, ground-truth tool) pair from the training dataset is treated as one training example, so a query with multiple ground-truth tools contributes multiple pairs. 

\paragraph{Training and early stopping procedure}
We use AdamW with learning rate $2 \times 10^{-5}$. We train for $20$ epochs with batch size $64$. After each epoch we evaluate inner-product retrieval on the held-out validation split and track Comp@5. Only the best Comp@$5$ epoch is retained. It should be noted that throughout all datasets, the best Comp@$5$ epoch was always retained before reaching 20 epochs. 

\subsection{\textsc{MMR}}

Maximal Margin Relevance (MMR) \citep{Carbonell1998MMR} is a retrieval method that selects a set of documents which are relevant to a given query while remaining mutually diverse, thus achieving broader information coverage unattainable with only semantic similarity scoring, such as dense retrieval. 
This is done by an iterative selection process given in Algorithm~\ref{alg:mmr}, where the score of each document is penalized by its similarity with documents already selected during subsequent iterations. The penalization is controlled by the parameter $\lambda \in [0,1]$, where larger $\lambda$ induces stronger penalty and thus larger diversity. 
We tune $\lambda$ over $\{0.0, 0.1, 0.2, 0.3, \ldots, 0.9, 1.0\}$, selecting the value which has the highest Comp@$5$ on the held-out validation set.  

\begin{algorithm}[t]
\caption{Maximal Marginal Relevance (MMR) Re-ranking}
\label{alg:mmr}
\begin{algorithmic}[1]
\Require Query embedding ${v} \in \mathbb{R}^{d}$; corpus embeddings $\{u_i\}_{i=1}^{N} \subset \mathbb{R}^{d}$; trade-off parameter $\lambda \in [0,1]$; number of items to select $K$.
\State $R \gets \varnothing$.
\For{$k = 1, \ldots, K$}
    \If{$k = 1$}
        \State $d_k \gets \arg\max_{i \in [N]}\ v^\top u_i$
    \Else
        \State $d_k \gets \arg\max_{i \in R^c} \Big[\, \lambda \ v^\top u_i \;-\; (1-\lambda) \ \max_{j \in R}\ u_i^\top u_j \,\Big]$
    \EndIf
    \State $R \gets R \cup \{d_k\}$
\EndFor
\State \Return $R$
\end{algorithmic}
\end{algorithm}

\subsection{\textsc{COLT}}
COLT~\citep{Qu2024CompletenessOriented} is a completeness-oriented tool retriever that fine-tunes
a bi-encoder based on a graph-based collaborative-learning: tools, scenes
(query intents), and queries are arranged in a heterogeneous graph, and a GNN refines the embeddings to encourage co-occurring tools to be retrieved together. We use the authors' official implementation\footnote{\url{https://github.com/quchangle1/COLT}}
without modification. Following the released configuration, we train for $20$ epochs with batch size $256$, learning rate $10^{-3}$, $\ell_2$ regularization $10^{-7}$, $10$ negatives per positive, and contrastive temperature $0.1$.
We sweep the GNN depth over $\{1, 2, 3\}$ and the dual-loss balancing
coefficient ($\lambda$) over $\{0.01, 0.03, 0.06, 0.1, 0.3, 0.6, 1.0\}$, and select the best
configuration with respect to Comp@5 via validation. 
The collaborative-learning stage is initialized from the bi-encoder checkpoint in \textsc{dense}.

\subsection{\textsc{NNN-Tr}, \textsc{NNN-Fix} and their $\textsc{L1/L2}$ variants}
\paragraph{Architecture.}
Both encoders are initialized from the bi-encoder checkpoint produced by \textsc{dense}. The query encoder is fully fine-tuned, while the corpus encoder is frozen and equipped with a trainable two-layer MLP adapter on top of its output embeddings. The adapter has hidden dimension $768$, uses a GELU activation \citep{Hendrycks2016GeLU}, and combines the residual and adapted representations via a learnable scalar mixing coefficient $\alpha$ as $(1-\sigma(\alpha))\,x + \sigma(\alpha)\,\mathrm{MLP}(x)$, with $\alpha$ initialized to $-5$ so that training begins close to the frozen baseline. All embeddings are $\ell_2$-normalized before being passed to the NNN decoder. 

\paragraph{Architecture.}
Both encoders are initialized from the bi-encoder checkpoint produced by \textsc{dense}. The query encoder is fully fine-tuned, while the frozen corpus encoder is equipped with a trainable two-layer MLP adapter on top of its output embeddings. 
The adapter has hidden dimension $768$, uses a GELU activation, and combines the residual and adapted representations via a learnable scalar mixing coefficient $\alpha$ as $(1-\sigma(\alpha))\,x + \sigma(\alpha)\,\mathrm{MLP}(x)$, with $\alpha$ initialized to $-5$ so that training begins close to the frozen baseline. 
This adaptation was taken since backpropagation through $T$ unrolled FISTA iterations requires storing the intermediate iterates $\{w^{(t)}, z^{(t)}\}_{t=1}^{T}$ over the full corpus matrix and incurs $O(dNT)$ activation memory per query, and therefore fine-tuning the full corpus encoder can be highly intensive computationally. 
All embeddings are $\ell_2$-normalized before being passed to the sparse-coding head.

\paragraph{Hyperparameters.}
Throughout all experiments, the hyperparameters $\gamma$ and $\tau$ included in the loss function \eqref{eq:training-objective} are fixed to $\gamma = 1.5$ and $\tau = 0.1$. 
On the other hand, $\lambda_1$ and $\lambda_2$ are each swept over $\{0.01,\, 0.03,\, 0.06,\, 0.1,\, 0.3,\, 0.6,\, 1.0\}$.

\paragraph{Training and early stopping procedure.}
We use AdamW with learning rate $2 \times 10^{-5}$, weight decay $0.01$ as our optimizer. With batch size $64$, we train on 20 epochs for the NumpyBank, ToolLens, and MultiHop-RAG datasets, while we train on 5 epochs for PandasBank and AWSBank. 
After each epoch, we evaluate NNN decoding (FISTA with $T = 50$) on the held-out validation split and track Comp@$5$. Only the best Comp@$5$ epoch is retained. Training is halted early if Comp@5 fails to strictly improve for three consecutive epochs (patience $= 3$), or if the FISTA solver returns a degenerate all-zero solution.

\section{Additional Numerical Experiments}
\label{sec:additional_exp}
In this section, we report additional numerical experiments for \textsc{dense}, \textsc{COLT}, \textsc{NNN-Fix}, \textsc{NNN-Tr}, and their variants under different backbone encoders. Specifically, we evaluate two further bi-encoders: \texttt{MiniLM-L6-cos-v5}\footnote{\url{https://huggingface.co/sentence-transformers/msmarco-MiniLM-L6-cos-v5}} and \texttt{distilbert-base-tas-b}\footnote{\url{https://huggingface.co/sentence-transformers/msmarco-distilbert-base-tas-b}}. Recall and completeness metrics for \texttt{MiniLM-L6-cos-v5} are reported in Tables~\ref{tab:recall-results_L6-cos-v5} and~\ref{tab:comp-results_L6-cos-v5}, and the corresponding results for \texttt{distilbert-base-tas-b} are reported in Tables~\ref{tab:recall-results_tasb} and~\ref{tab:comp-results_tasb}. For both backbone models, we can see that \textsc{NNN-Fix} can perform better than \textsc{dense}, \textsc{MMR} and \textsc{COLT} for the ToolBank datasets, while \textsc{NNN-Tr} can outperform all baseline methods across every dataset and metric, highlighting the effectiveness of our method.

\begin{table*}
\centering
\small
\setlength{\tabcolsep}{3.3pt}
\caption{Test-set Recall3 and Recall@5 (\%) per dataset using the \texttt{MiniLM-L6-cos-v5} backbone. Best per column is \textbf{bold}.}
\begin{tabular*}{\textwidth}{@{\extracolsep{\fill}}lllllllllll@{}}
\toprule
Method & \multicolumn{2}{c}{NumpyBank} & \multicolumn{2}{c}{PandasBank} & \multicolumn{2}{c}{AWSBank} & \multicolumn{2}{c}{ToolLens} & \multicolumn{2}{c}{MultiHop-RAG} \\
\cmidrule(lr){2-3}\cmidrule(lr){4-5}\cmidrule(lr){6-7}\cmidrule(lr){8-9}\cmidrule(lr){10-11}
 & R@3 & R@5 & R@3 & R@5 & R@3 & R@5 & R@3 & R@5 & R@3 & R@5 \\
\midrule
\textsc{dense} & 63.7 & 76.6 & 37.7 & 46.1 & 59.2 & 71.8 & 80.0 & 90.9 & 80.5 & 90.8 \\
\textsc{MMR}
 & \cg{66.5}{+2.8} & \cg{77.9}{+1.3}
 & \cg{42.1}{+4.4} & \cg{49.3}{+3.2}
 & \cg{65.1}{+5.9} & \cg{74.8}{+3.0}
 & \cg{83.9}{+3.8} & \cg{93.8}{+2.9}
 & \cg{80.5}{+0.0} & \cg{90.8}{+0.0} \\
\textsc{COLT}
 & \cg{64.0}{+0.3} & \cg{76.1}{-0.5}
 & \cg{38.5}{+0.8} & \cg{47.6}{+1.5}
 & \cg{59.2}{+0.0} & \cg{71.8}{+0.0}
 & \cg{87.8}{+7.8} & \cg{96.2}{+5.3}
 & \cg{80.9}{+0.4} & \cg{91.0}{+0.2} \\
\midrule
\textsc{NNN-Fix}
 & \cg{71.2}{+7.5} & \cg{82.1}{+5.5}
 & \cg{44.6}{+6.9} & \cg{52.9}{+6.8}
 & \cg{69.4}{+10.2}& \cg{78.8}{+7.0}
 & \cg{83.8}{+3.8} & \cg{94.4}{+3.5}
 & \cg{79.4}{-1.1} & \cg{90.8}{+0.0} \\
\textsc{\hspace{1.1em}L1-Fix}
 & \cg{71.0}{+7.3} & \cg{81.6}{+5.0}
 & \cg{44.6}{+6.9} & \cg{52.9}{+6.8}
 & \cg{69.4}{+10.2}& \cg{78.8}{+7.0}
 & \cg{84.8}{+4.8} & \cg{92.6}{+1.7}
 & \cg{73.7}{-6.8} & \cg{80.2}{-10.6} \\
\textsc{\hspace{1.1em}L2-Fix}
 & \cg{69.2}{+5.5} & \cg{80.0}{+3.4}
 & \cg{39.1}{+1.4} & \cg{47.7}{+1.6}
 & \cg{60.9}{+1.7} & \cg{73.5}{+1.7}
 & \cg{82.6}{+2.6} & \cg{93.1}{+2.2}
 & \cg{76.0}{-4.5} & \cg{88.9}{-1.9} \\
\midrule
\textsc{NNN-Tr}
 & \cg{73.3}{+9.6}        & \cg{\textbf{84.3}}{+7.7}
 & \cg{\textbf{46.8}}{+9.1}& \cg{55.5}{+9.4}
 & \cg{\textbf{71.8}}{+12.6}& \cg{\textbf{80.5}}{+8.7}
 & \cg{\textbf{95.7}}{+15.7}& \cg{\textbf{98.4}}{+7.5}
 & \cg{84.4}{+3.9}        & \cg{\textbf{92.1}}{+1.3} \\
\textsc{\hspace{1.1em}L1-Tr}
 & \cg{\textbf{73.5}}{+9.8}& \cg{83.2}{+6.6}
 & \cg{\textbf{46.8}}{+9.1}& \cg{\textbf{55.7}}{+9.6}
 & \cg{71.7}{+12.5}        & \cg{80.3}{+8.5}
 & \cg{95.2}{+15.2}        & \cg{97.7}{+6.8}
 & \cg{\textbf{85.4}}{+4.9}& \cg{\textbf{92.1}}{+1.3} \\
\textsc{\hspace{1.1em}L2-Tr}
 & \cg{\textbf{73.5}}{+9.8}& \cg{83.4}{+6.8}
 & \cg{46.5}{+8.8}        & \cg{55.5}{+9.4}
 & \cg{71.6}{+12.4}        & \cg{80.3}{+8.5}
 & \cg{95.2}{+15.2}        & \cg{98.1}{+7.2}
 & \cg{80.7}{+0.2}        & \cg{89.9}{-0.9} \\
\bottomrule
\end{tabular*}
\label{tab:recall-results_L6-cos-v5}
\vspace*{12pt}

\centering
\small
\setlength{\tabcolsep}{3.3pt}
\caption{Test-set Comp@3 and Comp@5 (\%) per dataset using the \texttt{MiniLM-L6-cos-v5} backbone. Best per column is \textbf{bold}.}
\begin{tabular*}{\textwidth}{@{\extracolsep{\fill}}lllllllllll@{}}
\toprule
Method & \multicolumn{2}{c}{NumpyBank} & \multicolumn{2}{c}{PandasBank} & \multicolumn{2}{c}{AWSBank} & \multicolumn{2}{c}{ToolLens} & \multicolumn{2}{c}{MultiHop-RAG} \\
\cmidrule(lr){2-3}\cmidrule(lr){4-5}\cmidrule(lr){6-7}\cmidrule(lr){8-9}\cmidrule(lr){10-11}
 & C@3 & C@5 & C@3 & C@5 & C@3 & C@5 & C@3 & C@5 & C@3 & C@5 \\
\midrule
\textsc{dense} & 25.9 & 45.3 & 6.8 & 12.6 & 26.0 & 40.2 & 53.3 & 78.3 & 56.0 & 79.6 \\
\textsc{MMR}
 & \cg{30.7}{+4.8} & \cg{47.8}{+2.5}
 & \cg{10.5}{+3.7}  & \cg{15.0}{+2.4}
 & \cg{33.4}{+7.4} & \cg{46.1}{+5.9}
 & \cg{64.1}{+10.8}& \cg{86.1}{+7.8}
 & \cg{56.0}{+0.0} & \cg{79.6}{+0.0} \\
\textsc{COLT}
 & \cg{26.8}{+0.9} & \cg{44.9}{-0.4}
 & \cg{6.6}{+0.2}  & \cg{12.6}{+0.0}
 & \cg{26.0}{+0.0} & \cg{40.7}{+0.5}
 & \cg{68.3}{+15.0}& \cg{90.0}{+11.7}
 & \cg{57.8}{+1.8} & \cg{\textbf{82.7}}{+3.1} \\
\midrule
\textsc{NNN-Fix}
 & \cg{36.6}{+10.7}& \cg{57.2}{+11.9}
 & \cg{11.1}{+5.2} & \cg{17.8}{+5.2}
 & \cg{36.7}{+10.7}& \cg{52.7}{+12.5}
 & \cg{62.1}{+8.8} & \cg{87.0}{+8.7}
 & \cg{52.0}{-4.0} & \cg{80.0}{+0.4} \\
\textsc{\hspace{1.1em}L1-Fix}
 & \cg{36.5}{+10.6}& \cg{56.2}{+10.9}
 & \cg{11.1}{+5.2} & \cg{17.8}{+5.2}
 & \cg{36.7}{+10.7}& \cg{52.8}{+12.6}
 & \cg{65.2}{+11.9}& \cg{83.4}{+5.1}
 & \cg{40.9}{-15.1}& \cg{52.0}{-27.6} \\
\textsc{\hspace{1.1em}L2-Fix}
 & \cg{33.1}{+7.2} & \cg{51.5}{+6.2}
 & \cg{6.2}{+0.3}  & \cg{11.9}{-0.7}
 & \cg{25.2}{-0.8} & \cg{41.1}{+0.9}
 & \cg{58.9}{+5.6} & \cg{84.0}{+5.7}
 & \cg{44.4}{-11.6}& \cg{76.0}{-3.6} \\
\midrule
\textsc{NNN-Tr}
 & \cg{39.1}{+13.2}        & \cg{\textbf{60.9}}{+15.6}
 & \cg{\textbf{12.9}}{+6.0}& \cg{20.6}{+8.0}
 & \cg{\textbf{39.4}}{+13.4}& \cg{\textbf{55.5}}{+15.3}
 & \cg{\textbf{90.7}}{+37.4}& \cg{\textbf{96.7}}{+18.4}
 & \cg{61.3}{+5.3}        & \cg{\textbf{82.7}}{+3.1} \\
\textsc{\hspace{1.1em}L1-Tr}
 & \cg{\textbf{40.8}}{+14.9}& \cg{58.9}{+13.6}
 & \cg{12.9}{+6.0}        & \cg{\textbf{20.8}}{+8.2}
 & \cg{39.3}{+13.3}        & \cg{55.1}{+14.9}
 & \cg{89.9}{+36.6}        & \cg{95.3}{+17.0}
 & \cg{\textbf{63.6}}{+7.6}& \cg{\textbf{82.7}}{+3.1} \\
\textsc{\hspace{1.1em}L2-Tr}
 & \cg{40.1}{+14.2}        & \cg{59.5}{+14.2}
 & \cg{11.9}{+5.0}         & \cg{20.1}{+7.5}
 & \cg{39.0}{+13.0}        & \cg{55.0}{+14.8}
 & \cg{89.6}{+36.3}        & \cg{96.3}{+18.0}
 & \cg{55.6}{-0.4}         & \cg{79.1}{-0.5} \\
\bottomrule
\end{tabular*}
\label{tab:comp-results_L6-cos-v5}
\vspace*{12pt}
\end{table*}
\begin{table*}

\centering
\small
\setlength{\tabcolsep}{3.3pt}
\caption{Test-set Recall@3 and Recall@5 (\%) per dataset using the \texttt{distilbert-base-tas-b} backbone. Best per column is \textbf{bold}.}
\begin{tabular*}{\textwidth}{@{\extracolsep{\fill}}lllllllllll@{}}
\toprule
Method & \multicolumn{2}{c}{NumpyBank} & \multicolumn{2}{c}{PandasBank} & \multicolumn{2}{c}{AWSBank} & \multicolumn{2}{c}{ToolLens} & \multicolumn{2}{c}{MultiHop-RAG} \\
\cmidrule(lr){2-3}\cmidrule(lr){4-5}\cmidrule(lr){6-7}\cmidrule(lr){8-9}\cmidrule(lr){10-11}
 & R@3 & R@5 & R@3 & R@5 & R@3 & R@5 & R@3 & R@5 & R@3 & R@5 \\
\midrule
\textsc{dense} & 68.6 & 81.0 & 42.2 & 51.1 & 65.8 & 77.0 & 86.1 & 96.5 & 83.6 & 92.0 \\
\textsc{MMR}
  & \cg{72.7}{+4.0} & \cg{81.8}{+0.7}
 & \cg{44.5}{+2.3} & \cg{52.8}{+1.7}
 & \cg{70.0}{+4.2} & \cg{78.8}{+1.8}
 & \cg{90.0}{+3.9} & \cg{97.3}{+0.8}
 & \cg{83.6}{+0.0} & \cg{92.0}{+0.0} \\
\textsc{COLT}
 & \cg{68.7}{+0.1} & \cg{81.1}{+0.1}
 & \cg{42.1}{-0.1} & \cg{51.2}{+0.1}
 & \cg{65.8}{+0.0} & \cg{77.0}{+0.0}
 & \cg{87.8}{+1.7} & \cg{97.2}{+0.7}
 & \cg{83.7}{+0.1} & \cg{92.0}{+0.0} \\
\midrule
\textsc{NNN-Fix}
 & \cg{74.6}{+6.0} & \cg{85.2}{+4.2}
 & \cg{48.3}{+6.1} & \cg{56.8}{+5.7}
 & \cg{73.1}{+7.3} & \cg{81.6}{+4.6}
 & \cg{91.1}{+5.0} & \cg{97.6}{+1.1}
 & \cg{81.9}{-1.7} & \cg{92.7}{+0.7} \\
\textsc{\hspace{1.1em}L1-Fix}
 & \cg{74.6}{+6.0} & \cg{84.8}{+3.8}
 & \cg{48.2}{+6.0} & \cg{56.9}{+5.8}
 & \cg{73.0}{+7.2} & \cg{81.4}{+4.4}
 & \cg{92.6}{+6.5} & \cg{96.8}{+0.3}
 & \cg{77.8}{-5.8} & \cg{85.3}{-6.7} \\
\textsc{\hspace{1.1em}L2-Fix}
 & \cg{73.6}{+5.0} & \cg{84.0}{+3.0}
 & \cg{43.8}{+1.6} & \cg{53.5}{+2.4}
 & \cg{69.6}{+3.8} & \cg{79.5}{+2.5}
 & \cg{92.4}{+6.3} & \cg{97.7}{+1.2}
 & \cg{81.7}{-1.9} & \cg{92.4}{+0.4} \\
\midrule
\textsc{NNN-Tr}
 & \cg{\textbf{76.1}}{+7.5}& \cg{\textbf{85.7}}{+4.7}
 & \cg{\textbf{49.8}}{+7.6}& \cg{\textbf{58.2}}{+7.1}
 & \cg{\textbf{74.1}}{+8.3}& \cg{\textbf{82.5}}{+5.5}
 & \cg{\textbf{96.5}}{+10.4}& \cg{\textbf{98.4}}{+1.9}
 & \cg{86.1}{+2.5}        & \cg{\textbf{93.8}}{+1.8} \\
\textsc{\hspace{1.1em}L1-Tr}
 & \cg{75.5}{+6.9}        & \cg{84.9}{+3.9}
 & \cg{49.7}{+7.5}        & \cg{58.0}{+6.9}
 & \cg{\textbf{74.1}}{+8.3}& \cg{82.3}{+5.3}
 & \cg{96.3}{+10.2}       & \cg{98.3}{+1.8}
 & \cg{\textbf{87.0}}{+3.4}& \cg{92.7}{+0.7} \\
\textsc{\hspace{1.1em}L2-Tr}
 & \cg{75.1}{+6.5}        & \cg{85.1}{+4.1}
 & \cg{47.8}{+5.6}        & \cg{56.9}{+5.8}
 & \cg{73.0}{+7.2}        & \cg{81.8}{+4.8}
 & \cg{96.3}{+10.2}       & \cg{\textbf{98.4}}{+1.9}
 & \cg{84.5}{+0.9}        & \cg{92.0}{+0.0} \\
\bottomrule
\end{tabular*}
\label{tab:recall-results_tasb}
\vspace*{12pt}
\centering
\small
\setlength{\tabcolsep}{3.3pt}
\caption{Test-set Comp@3 and Comp@5 (\%) per dataset using the \texttt{distilbert-base-tas-b} backbone. Best per column is \textbf{bold}.}
\begin{tabular*}{\textwidth}{@{\extracolsep{\fill}}lllllllllll@{}}
\toprule
Method & \multicolumn{2}{c}{NumpyBank} & \multicolumn{2}{c}{PandasBank} & \multicolumn{2}{c}{AWSBank} & \multicolumn{2}{c}{ToolLens} & \multicolumn{2}{c}{MultiHop-RAG} \\
\cmidrule(lr){2-3}\cmidrule(lr){4-5}\cmidrule(lr){6-7}\cmidrule(lr){8-9}\cmidrule(lr){10-11}
 & C@3 & C@5 & C@3 & C@5 & C@3 & C@5 & C@3 & C@5 & C@3 & C@5 \\
\midrule
\textsc{dense} & 33.0 & 53.4 & 9.0 & 15.7 & 32.5 & 48.1 & 65.3 & 92.3 & 60.0 & 82.2 \\
\textsc{MMR}
 & \cg{40.6}{+7.6} & \cg{56.5}{+3.0}
 & \cg{11.2}{+2.2}  & \cg{17.8}{+2.1}
 & \cg{38.0}{+5.5} & \cg{52.0}{+3.9}
 & \cg{76.4}{+11.1} & \cg{94.7}{+2.4}
 & \cg{60.0}{+0.0} & \cg{82.2}{+0.0} \\
\textsc{COLT}
 & \cg{33.5}{+0.5} & \cg{53.8}{+0.4}
 & \cg{8.8}{-0.2}  & \cg{15.7}{+0.0}
 & \cg{32.5}{+0.0} & \cg{48.1}{+0.0}
 & \cg{67.6}{+2.3} & \cg{92.7}{+0.4}
 & \cg{60.0}{+0.0} & \cg{82.7}{+0.5} \\
\midrule
\textsc{NNN-Fix}
 & \cg{42.4}{+9.4} & \cg{63.2}{+9.8}
 & \cg{13.8}{+4.8} & \cg{21.8}{+6.1}
 & \cg{41.0}{+8.5} & \cg{57.3}{+9.2}
 & \cg{79.2}{+13.9}& \cg{95.4}{+3.1}
 & \cg{57.8}{-2.2} & \cg{83.6}{+1.4} \\
\textsc{\hspace{1.1em}L1-Fix}
 & \cg{42.5}{+9.5} & \cg{62.8}{+9.4}
 & \cg{13.8}{+4.8} & \cg{21.9}{+6.2}
 & \cg{41.0}{+8.5} & \cg{57.1}{+9.0}
 & \cg{83.5}{+18.2}& \cg{93.2}{+0.9}
 & \cg{47.6}{-12.4}& \cg{64.4}{-17.8} \\
\textsc{\hspace{1.1em}L2-Fix}
 & \cg{40.2}{+7.2} & \cg{60.2}{+6.8}
 & \cg{8.6}{-0.4}  & \cg{16.2}{+0.5}
 & \cg{35.4}{+2.9} & \cg{51.9}{+3.8}
 & \cg{82.8}{+17.5}& \cg{95.3}{+3.0}
 & \cg{57.3}{-2.7} & \cg{83.1}{+0.9} \\
\midrule
\textsc{NNN-Tr}
 & \cg{\textbf{44.6}}{+11.6}& \cg{\textbf{64.0}}{+10.6}
 & \cg{14.9}{+5.9}         & \cg{23.1}{+7.4}
 & \cg{42.0}{+9.5}         & \cg{\textbf{59.0}}{+10.9}
 & \cg{\textbf{93.0}}{+27.7}& \cg{\textbf{97.1}}{+4.8}
 & \cg{66.7}{+6.7}         & \cg{\textbf{85.8}}{+3.6} \\
\textsc{\hspace{1.1em}L1-Tr}
 & \cg{44.1}{+11.1}        & \cg{62.5}{+9.1}
 & \cg{\textbf{15.3}}{+6.3}& \cg{\textbf{23.6}}{+7.9}
 & \cg{\textbf{42.2}}{+9.7}& \cg{58.7}{+10.6}
 & \cg{92.4}{+27.1}        & \cg{96.7}{+4.4}
 & \cg{\textbf{67.6}}{+7.6}& \cg{84.0}{+1.8} \\
\textsc{\hspace{1.1em}L2-Tr}
 & \cg{43.2}{+10.2}        & \cg{62.5}{+9.1}
 & \cg{12.5}{+3.5}         & \cg{20.6}{+4.9}
 & \cg{40.6}{+8.1}         & \cg{56.9}{+8.8}
 & \cg{92.3}{+27.0}        & \cg{\textbf{97.1}}{+4.8}
 & \cg{64.9}{+4.9}         & \cg{83.1}{+0.9} \\
\bottomrule
\end{tabular*}
\label{tab:comp-results_tasb}
\end{table*}

\end{document}